\documentclass{bmcart}

\usepackage{amsmath,mathptmx,amssymb,graphicx,amsthm} 
\usepackage[english]{babel}
\usepackage[utf8]{inputenc} %unicode support
\newtheorem{theorem}{Theorem}
\newtheorem{remark}{Remark}
\newtheorem{corollary}{Corollary}

\startlocaldefs
\newcommand{\ave}[1]{\mathbb{E}( #1)}
\endlocaldefs

\begin{document}

%%% Start of article front matter
\begin{frontmatter}

\begin{fmbox}
\dochead{Research}


\title{A monotonic relationship between the variability of the infectious period and final size in pairwise epidemic modelling}

%%%%%%%%%%%%%%%%%%%%%%%%%%%%%%%%%%%%%%%%%%%%%%
%%                                          %%
%% Enter the authors here                   %%
%%                                          %%
%% Specify information, if available,       %%
%% in the form:                             %%
%%   <key>={<id1>,<id2>}                    %%
%%   <key>=                                 %%
%% Comment or delete the keys which are     %%
%% not used. Repeat \author command as much %%
%% as required.                             %%
%%                                          %%
%%%%%%%%%%%%%%%%%%%%%%%%%%%%%%%%%%%%%%%%%%%%%%

\author[
   addressref={aff1},                   % id's of addresses, e.g. {aff1,aff2}
                        % id of corresponding address, if any
  % noteref={n1},                        % id's of article notes, if any
   email={rost@math.u-szeged.hu}   % email address
]{\inits{JE}\fnm{Zsolt} \snm{Vizi}}
\author[
   addressref={aff2},
   email={rost@math.u-szeged.hu}
]{\inits{JRS}\fnm{István Z.} \snm{Kiss}}
\author[
addressref={aff3},
email={rost@math.u-szeged.huk}
]{\inits{JRS}\fnm{Joel C.} \snm{Miller}}
\author[
addressref={aff1,aff4},
corref={aff1}, 
email={rost@math.u-szeged.hu}
]{\inits{JRS}\fnm{Gergely} \snm{R\"ost}}

%%%%%%%%%%%%%%%%%%%%%%%%%%%%%%%%%%%%%%%%%%%%%%
%%                                          %%
%% Enter the authors' addresses here        %%
%%                                          %%
%% Repeat \address commands as much as      %%
%% required.                                %%
%%                                          %%
%%%%%%%%%%%%%%%%%%%%%%%%%%%%%%%%%%%%%%%%%%%%%%

\address[id=aff1]{%                           % unique id
  \orgname{Bolyai Institute, University of Szeged}, % university, etc
  \street{Aradi v\'ertan\'uk tere 1}                     %
  \postcode{H-6720}                                % post or zip code
  \city{Szeged},                              % city
  \cny{Hungary}                                    % country
}
\address[id=aff2]{%
  \orgname{University of Sussex},
  \street{School of Mathematical and Physical Sciences, Department of Mathematics, University of Sussex,},
  \postcode{BN1 9QH}
  \city{Falmer, Brighton},
  \cny{UK}
}
\address[id=aff3]{%
	\orgname{Institute for Disease Modeling},
	\street{}
	\postcode{}
	\city{Bellevue, Washington},
	\cny{USA}
}
\address[id=aff4]{%
	\orgname{Mathematical Institute, University of Oxford},
	\street{Woodstock Road},
	\postcode{OX2 6GG}
	\city{Oxford},
	\cny{United Kingdom}
}

%%%%%%%%%%%%%%%%%%%%%%%%%%%%%%%%%%%%%%%%%%%%%%
%%                                          %%
%% Enter short notes here                   %%
%%                                          %%
%% Short notes will be after addresses      %%
%% on first page.                           %%
%%                                          %%
%%%%%%%%%%%%%%%%%%%%%%%%%%%%%%%%%%%%%%%%%%%%%%

\begin{artnotes}
%\note{Sample of title note}     % note to the article
%\note[id=n1]{Equal contributor} % note, connected to author
\end{artnotes}

\end{fmbox}% comment this for two column layout

%%%%%%%%%%%%%%%%%%%%%%%%%%%%%%%%%%%%%%%%%%%%%%
%%                                          %%
%% The Abstract begins here                 %%
%%                                          %%
%% Please refer to the Instructions for     %%
%% authors on http://www.biomedcentral.com  %%
%% and include the section headings         %%
%% accordingly for your article type.       %%
%%                                          %%
%%%%%%%%%%%%%%%%%%%%%%%%%%%%%%%%%%%%%%%%%%%%%%

\begin{abstractbox}

\begin{abstract} 
	 For a recently derived pairwise model of network epidemics with non-Markovian recovery,  we prove that under some mild technical conditions on the distribution of the infectious periods,  smaller variance in the recovery time leads to higher reproduction number, and consequently to a larger epidemic outbreak, when the mean infectious period is fixed. We discuss how this result is related to various stochastic orderings of the distributions of infectious periods. The results are illustrated by a number of explicit stochastic simulations, suggesting that their validity goes beyond regular networks. 
	
	% abstract
%\parttitle{First part title} %if any
%Text for this section.

%\ parttitle{Second part title} %if any
%Text for this section.
\end{abstract}

%%%%%%%%%%%%%%%%%%%%%%%%%%%%%%%%%%%%%%%%%%%%%%
%%                                          %%
%% The keywords begin here                  %%
%%                                          %%
%% Put each keyword in separate \kwd{}.     %%
%%                                          %%
%%%%%%%%%%%%%%%%%%%%%%%%%%%%%%%%%%%%%%%%%%%%%%


% MSC classifications codes, if any
%\begin{keyword}[class=AMS]
%\kwd[Primary ]{}
%\kwd{}
%\kwd[; secondary ]{}
%\end{keyword}

\end{abstractbox}
%
%\end{fmbox}% uncomment this for twcolumn layout

\end{frontmatter}

%%%%%%%%%%%%%%%%%%%%%%%%%%%%%%%%%%%%%%%%%%%%%%
%%                                          %%
%% The Main Body begins here                %%
%%                                          %%
%% Please refer to the instructions for     %%
%% authors on:                              %%
%% http://www.biomedcentral.com/info/authors%%
%% and include the section headings         %%
%% accordingly for your article type.       %%
%%                                          %%
%% See the Results and Discussion section   %%
%% for details on how to create sub-sections%%
%%                                          %%
%% use \cite{...} to cite references        %%
%%  \cite{koon} and                         %%
%%  \cite{oreg,khar,zvai,xjon,schn,pond}    %%
%%  \nocite{smith,marg,hunn,advi,koha,mouse}%%
%%                                          %%
%%%%%%%%%%%%%%%%%%%%%%%%%%%%%%%%%%%%%%%%%%%%%%

%%%%%%%%%%%%%%%%%%%%%%%%% start of article main body
% <put your article body there>

%%%%%%%%%%%%%%%%
%% Background %%
%%
%\section*{Content}
%Text and results for this section, as per the individual journal's instructions for authors. %\cite{koon,oreg,khar,zvai,xjon,schn,pond,smith,marg,hunn,advi,koha,mouse}

\section*{Introduction}

\medskip

Networks provide a useful paradigm to incorporate contact patterns and various heterogeneities within a population \cite{pastor2014epidemic,newman2002spread,konyv}. The basic ingredients of such models are nodes and links, usually representing individuals and the contacts between them, but they may represent also groups of individuals (such as the population at some geographic location), and the connectedness of these groups (such as transportation routes \cite{dia,yuki}).
In simple disease outbreak models, the status of an individual can be susceptible ($S$), infected ($I$) or recovered ($R$).  A key parameter associated with most epidemic models is the basic reproduction number (denoted by $\mathcal{R}_0$), which denotes the expected number of secondary infections generated by a typical infected individual introduced into a fully susceptible population \cite{diekmann}. The reproduction number is also a threshold quantity: if $
\mathcal{R}_0<1$ the epidemic will die out, while if $\mathcal{R}_0>1$ the disease may spread. Another important measure of epidemic severity is the final epidemic size, which is the total 
number of individuals who become infected during the time course of the epidemic. These two quantities are often connected via the so-called final size relation.  In these simple models that assume a fully mixed population, the final fraction that is not infected $s_\infty$ solves the implicit relation
\[
s_\infty = S(0) e^{-\mathcal{R}_0 (1-s_\infty)} \, .
\]
If infected individuals transmit with constant rate $\beta$, then in this well-mixed model $\mathcal{R}_0 = \beta \ave{\mathcal{I}}$ where $\ave{\mathcal{I}}$ is the average infection duration, and so variance in the distribution of infection duration does not affect the final size~\cite{miller:final,ma}.

Modelling epidemics on networks however increases the complexity of the models since the underlying population structure means that individuals are not interchangeable.  Thus we must track which individuals are in each status rather than simply how many individuals are in each status.  For example, in the most fundamental case of Markovian transmission and recovery, both time to infection and the time spent as infected and infectious is taken from exponential distributions with appropriate rates. Even for the purely Markovian case we need to deal with a continuous time Markov chain with a discrete state space with $3^N$ elements, where three stands for the three possible states a node can be in ($S$, $I$ and $R$) and $N$ denotes the number of nodes in the network. Writing down evolution equations for the probability of the system being in any of these states is possible but impractical due to the high dimensionality of the system. Hence, in order to deal with this complexity one need to employ some `clever' averaging.

Probabilistic methods, such as branching processes can be used to deal with the early growth and the asymptotic behaviour \cite{Ball}, with percolation theory also leading to good analytical treatment for the early growth and final size \cite{Perc}.  For the later dynamics, we generally need to derive a mean-field model, e.g. a low dimensional system of ODEs.

 There are many well established ways to derive mean-field models. Perhaps the most compact method is the so called edge based compartmental model (EBCM) \cite{EBCM} which has been successfully used to capture SIR dynamics with arbitrary transmission and infection processes \cite{Neil} on configuration-like networks. The EBCM provides an excellent approximation of the exact stochastic network epidemic, which becomes exact in some appropriate limits and conditions on the underlying network \cite{Decreusefond,Janson}.

 Another powerful method to model epidemic spread on network is provided by the message passing approach \cite{MP} and this works for arbitrary transmission and recovery processes but at the expense of a system consisting of a large number of integro-differential equations. 

In addition, pairwise models have been successfully used to approximate stochastic epidemics on networks and represent a vast improvement on compartmental models. 
Pairwise models also have the advantage of being easy to understand and very intuitive when compared to the EBCM or the message passing model. 

All the above are able to capture the time evolution of the epidemic while also offering insights about the epidemic threshold and final size.  All these models have the same starting point and not surprisingly it can be shown that often these models are equivalent~\cite{MLA,Neil,konyv} and they simply represent different choices of how one averages and how the reduced state space is defined \cite{konyv}.  

While dealing with the complexity and the modelling of contact structures, the dynamics of the disease needs to be accounted for appropriately. It is well known that the duration of the infectiousness has a major impact on whether an outbreak happens and how many people it affects as being a key parameter in the basic reproduction number. To highlight a recent example, in the West-African ebola outbreak one crucial part of the intervention strategy was to reduce the length of the post-mortem infectious period \cite{barbarossa}. In this paper we bridge the gap by considering a model that can capture both the complexity of contact structure as well as the features of the disease itself. To do this we consider pairwise models with Markovian infection but arbitrary recovery process and we focus on the outbreak threshold derived from this model and its dependence on the choice of the recovery process. The paper is structured as follows. First, we introduce the pairwise model, the analytical final epidemic size relation followed by the newly introduced basic pairwise reproduction number $\mathcal{R}_0^p$. The main result of the paper is on the relation between the variance in the distribution of the recovery process and the basic pairwise reproduction number. This is followed by some discussion of our results with respect to the concept of stochastic ordering, and the possible extension of our results to heterogeneous networks. We conclude with extensive numerical results and a discussion of our findings.

\section*{Methods}

\medskip

Pairwise models are formulated in terms of the expected values for the number of susceptible ($[S]$), infected ($[I]$) and recovered ($[R]$) nodes, which depend on the expected values of ($SS$) pairs ($[SS]$) and ($SI$) pairs ($[SI]$).    
Introducing the usual notations
\begin{itemize}
	\item $[X](t)$ for the expected number of nodes in state $X$ at time $t$, 
	\item $[XY](t)$ for the expected number of links connecting a node in state $X$ to another in state $Y$, and 
	\item $[XYZ](t)$ for the expected number of triplets in state $X-Y-Z$, 
\end{itemize}
where, $X, Y, Z\in \{S, I, R\}$,
and by summing up all possible transitions, the pairwise model reads as
\begin{eqnarray}
	\dot{[S]}(t)&=&-\tau [SI](t),\nonumber \\
	\dot{[I]}(t)&=&\tau [SI](t)-\gamma [I](t), \nonumber \\
	%\dot{[R]}(t)&=&\gamma [I](t),\nonumber \\
	\dot{[SS]}(t)&=&-2\tau [SSI](t), \label{Original_Pairwise}\\
	\dot{[SI]}(t)&=&\tau [SSI](t)-\tau [ISI](t)-\tau [SI](t)-\gamma [SI](t),\nonumber
	%\dot{[SR]}(t)&=&-\tau [ISR](t) + \gamma [SI](t),\nonumber \\
	%\dot{[II]}(t)&=&2\tau [ISI](t) + 2\tau [SI](t)-2\gamma [II](t), \nonumber \\
	%\dot{[IR]}(t)&=&\tau [ISR](t)+\gamma [II](t)-\gamma [IR](t),\nonumber \\
	%\dot{[RR]}(t)&=&\gamma [IR](t)\nonumber,
\end{eqnarray}
where $\tau$ is the per contact infection rate and $\gamma$ is the recovery rate. 
Here $[S]+[I]+[R]=N$ is the total number of nodes in the network, and only those equations are listed which are necessary to derive a complete self-consistent system. The equations for links contain triplets, 
thus we have to break the dependence on higher order terms to obtain a closed system. The closure approximation formula 
$[XSY]=\frac{n-1}{n} \frac{[XS] [SY]}{[S]}$, where $n$ is the average number of links per node, leads to the self-consistent system \cite{keeling} 
\begin{eqnarray}
	\label{eq:pairmarkov}
	\dot{[S]}(t)&=&-\tau [SI](t),\nonumber \\
	\dot{[I]}(t)&=&\tau [SI](t)-\gamma [I](t),\nonumber \\
	\dot{[SS]}(t)&=&-2\tau \frac{n-1}{n} \frac{[SS](t)[SI](t)}{[S](t)}, \\
	\dot{[SI]}(t)&=&\tau \frac{n-1}{n} \left(\frac{[SS](t)[SI](t)}{[S](t)}-\frac{[SI](t)[SI](t)}{[S](t)}\right)-(\tau+\gamma)[SI](t)\nonumber.
\end{eqnarray}
Closing at the level of pairs with the approximation $[XY]=n[X]\frac{[Y]}{N}$, one obtains the so called mean-field model
(or compartmental model)
\begin{eqnarray}
	\dot{S}(t)&=&-\tau \frac{n}{N} S(t) I(t),\nonumber \\
	\dot{I}(t)&=&\tau \frac{n}{N} S(t) I(t)-\gamma I(t),
\end{eqnarray}
with basic reproduction number
\begin{equation}
	\mathcal{R}_0=\frac{n}{N}\tau \mathbb{E}(\mathcal{I}) S_{0},\label{eq:mfstandardR0}
\end{equation}
where, $\mathbb{E}(\mathcal{I})=1/\gamma$ is the expected infectious period. 
	The final size relation associated to the mean-field model is 
\begin{eqnarray}
\label{eq:finalsizemftheorem}
\ln\left(s_\infty\right)=\mathcal{R}_0\left(s_\infty-1\right),
\end{eqnarray}
where $S_0$ is the number of susceptible individuals at time $t=0$ and $s_\infty=S_{\infty}/S_0$, where $S({\infty})=S_{\infty}$.
%These models lead to the following $\mathcal{R}_0$ values:
%\begin{equation}
%\mathcal{R}_0^{pw}=\frac{(n-2)\tau}{\gamma}, \mathcal{R}_0^{mf}=\frac{\tau n}{\gamma}.
%\end{equation}
There are many results for the Markovian pairwise models \cite{saldana,keeling,konyv}, for example, the final epidemic size is given by 
\begin{equation}
	\frac{s_\infty^{\frac{1}{n}}-1}{\frac{1}{n-1}}=\frac{n-1}{N}\frac{\tau}{\tau+\gamma}[S]_0\left(s_\infty^{\frac{n-1}{n}} -1 \right),
    \label{eq:homogeneous_final}
\end{equation}
where $[S]_0$ is the number of susceptible individuals at time $t=0$ and $s_\infty=[S]_{\infty}/[S]_0$, where $[S]({\infty})=[S]_{\infty}$.

\subsection*{Non-Markovian Recovery}

\medskip

The Markovianity of the recovery process is a strong simplifying assumption. For many epidemics, the infectious period has great importance and it is measured empirically.  Recently, pairwise approximations of the SIR dynamics  with non-Markovian recovery have been derived, see \cite{prl,wilkinson,biomat,proca}. In the special case of fixed recovery time $\sigma$, the mean-field model is given by
\begin{eqnarray}
	\label{eq:meanfield}
	S'(t)&=&-\tau \frac{n}{N} S(t)I(t), \nonumber \\
	I'(t)&=&\tau \frac{n}{N} S(t)I(t)-\tau \frac{n}{N} S(t-\sigma)I(t-\sigma),
\end{eqnarray}
while the pairwise model turned out to be \cite{prl}
\begin{eqnarray}
	\label{eq:closeq}
	\dot{[S]}(t)&=&-\tau [SI](t),\nonumber \\
	\dot{[SS]}(t)&=&-2\tau \frac{n-1}{n} \frac{[SS](t) [SI](t)}{[S](t)},\nonumber \\
	\dot{[I]}(t)&=& \tau [SI](t) - \tau [SI](t-\sigma), \nonumber \\
	\dot{[SI]}(t)&=& \tau \frac{n-1}{n}\frac{[SS](t)[SI](t)}{[S](t)}-\tau \frac{n-1}{n}\frac{[SI](t)[SI](t)}{[S](t)} -\tau [SI](t) \nonumber \\
	& &-\tau \frac{n-1}{n}\frac{[SS](t-\sigma)[SI](t-\sigma)}{[S](t-\sigma)} e^{-\int_{t-\sigma}^{t}\tau\frac{n-1}{n}\frac{[SI](u)}{[S](u)}+\tau du}.
\end{eqnarray}
Both systems are now delay differential equations rather than ordinary differential equations, as is the case for Markovian epidemics. In \cite{prl}, the following final epidemic size relation has been derived:
\begin{equation}
	\frac{s_\infty^{\frac{1}{n}}-1}{\frac{1}{n-1}}=\frac{n-1}{N}\left(1-e^{-\tau \sigma}\right)[S]_0\left(s_\infty^{\frac{n-1}{n}}-1\right). \label{finalsize}
\end{equation}
 Considering a general distribution for the recovery period, the pairwise model can be formulated as a system of integro-differential equations \cite{proca,wilkinson}, which is given by 
\begin{subequations}
	\label{eq:closeq2}
	\begin{align}
	\dot{[S]}(t)&=-\tau [SI](t) \label{eq:closedeqS}\\
	\dot{[SS]}(t)&=-2\tau \frac{n-1}{n} \frac{[SS](t) [SI](t)}{[S](t)} \label{eq:closedeqSS}\\
	\dot{[I]}(t)&=\tau [SI](t) - \int_{0}^{t} \tau [SI](t-a) f(a) da - \int_{t}^{\infty} \varphi(a-t) \frac{f_\mathcal{I}(a)}{\xi(a-t)} da\label{eq:closedeqI}\\
	\dot{[SI]}(t)&=\tau \frac{n-1}{n}\frac{[SS](t)[SI](t)}{[S](t)}-\tau \frac{n-1}{n}\frac{[SI](t)}{[S](t)}[SI](t)-\tau [SI](t)\nonumber\\
	&-\int_0^t \tau \frac{n-1}{n}\frac{[SS](t-a)[SI](t-a)}{[S](t-a)} e^{-\int_{t-a}^t \tau\frac{n-1}{n}\frac{[SI](s)}{[S](s)}+\tau ds} f_\mathcal{I}(a)da\nonumber \\
	&-\int_t^{\infty}\frac{n}{N} [S]_0 \varphi (a-t) e^{-\int_{0}^t 
		\tau\frac{n-1}{n}\frac{[SI](s)}{[S](s)} +\tau ds} \frac{f_\mathcal{I}(a)}{\xi(a-t)}da. \label{eq:closedeqSI}
	\end{align}
\end{subequations}

Above we assume that the infection process along $S$--$I$ links is Markovian with transmission rate $\tau>0$. The recovery part is considered to be non-Markovian given by a random variable $\mathcal{I}$, with a cumulative distribution function $F_\mathcal{I}(a)$ and probability density function $f_\mathcal{I}(a)$. We use the associated survival function $\xi(a)=1-F_\mathcal{I}(a)$ and hazard function $h(a)=-\frac{\xi'(a)}{\xi(a)}=\frac{f(a)}{\xi(a)}$. We note that $\varphi (a)$ is the initial condition which gives the age of infection of individuals at time $t=0$.

From Eq.~\eqref{eq:closeq2}, the associated mean-field model can be easily deduced by using the closure approximation formula for homogeneous networks (i.e. $n$-regular graphs)
\begin{equation}
\label{eq:closeformmf}
[XY](t)=\frac{n}{N}[X](t)[Y](t),
\end{equation}
thus the node-level system becomes
\begin{subequations}
	\label{eq:closmfeq}
	\begin{align}
	\dot{S}(t)&=-\tau \frac{n}{N}S(t)I(t) \label{eq:closedmfeqS}\\
	\dot{I}(t)&=\tau \frac{n}{N}S(t)I(t) - \int_{0}^{t} \tau \frac{n}{N}S(t-a)I(t-a) f_\mathcal{I}(a) da- \int_{t}^{\infty} \varphi(a-t) \frac{f_\mathcal{I}(a)}{\xi(a-t)} da.\label{eq:closedmfeqI}
	\end{align}
\end{subequations}
%In the following, we investigate these systems from mathematical and numerical point of view, focussing on the epidemiologically meaningful properties of the models.

\section*{The Pairwise Reproduction Number and Infectious Times}

\medskip

In \cite{prl}, a newly introduced basic reproduction-like number is defined for fixed length infectious periods as 
\begin{equation}
\mathcal{R}_0^{p}:=\frac{n-1}{N}\left(1-e^{-\tau \sigma}\right)[S]_0,
\label{eq:fixed_time_R0}
\end{equation} 
which appears also in equation~\eqref{finalsize}. 
%We can interpret this as the probability of transmission $1-e^{-\tau \sigma}$ and the expected number of new $S$--$I$ edges created by a transmission $(n-1)[S]_0/N$.  So this is the expected number of $S$--$I$ edges created by an $S$--$I$ edge.
%
It has also been shown, that for arbitrary
infectious periods, the basic reproduction number of the pairwise model is
\begin{equation}
\mathcal{R}_0^p=\frac{n-1}{N}\left(1-\mathcal L[f_\mathcal{I}](\tau)\right)[S]_0,\label{eq:genR0p}
\end{equation}
where $\mathcal L[\cdot]$ is the Laplace transform and $f_\mathcal{I}$
is the probability density function of the recovery process given by
the random variable $\mathcal{I}$. 
%This has the same interpretation as Equation~\eqref{eq:fixed_time_R0}, with $1-\mathcal{L}[f_{\mathcal{I}}](\tau)$ being the transmission probability.
%
%
Numerical tests and analytical results have both confirmed that, in general, the 
following implicit relation for the final epidemic size holds
\begin{align}
\frac{s_\infty^{\frac{1}{n}}-1}{\frac{1}{n-1}}&=\mathcal{R}^p_0\left(s_\infty^{\frac{n-1}{n}}-1\right) =\frac{n-1}{N}\left(1-\mathcal L[f_\mathcal{I}](\tau)\right)[S]_0\left(s_\infty^{\frac{n-1}{n}}-1\right). \label{eq:finalsizegenimprel}
\end{align}

 Several important observations can be made. The first is around the interpretation of the Laplace transform of $f_\mathcal{I}$.
Let us consider an isolated $S$--$I$ link, and let $\mathcal{E}$ be the exponentially distributed random variable of the time of infection along this link, with parameter $\tau$. Then the probability of transmission is the same as the probability that infection occurs before recovery, that is
\begin{equation}T=P(\mathcal{E}<\mathcal{I})=\int_0^\infty F_\mathcal{E}(y) f_\mathcal{I}(y) dy=\int_0^\infty(1-e^{-\tau y}) f_\mathcal{I}(y) dy=1-\mathcal L[f_\mathcal{I}](\tau).\end{equation} 
Hence, the Laplace transform has natural interpretation and enters the calculation of the probability of transmission across an isolated $S$--$I$ link.
%Hence,
%\begin{equation}
%\mathcal{R}^p_0=\frac{n-1}{N} [S]_0\left(1-\mathcal L[f_\mathcal{I}](\tau)\right) \underbrace{=}_{[S]_0 \rightarrow N}\left(1-\mathcal L[f](\tau)\right)\left(\frac{n^2-n}{n}\right)=\mathcal{R}_0^s,
%%\leftrigharrow \mathcal{R}^p_0=\frac{n-1}{N} [S]_0\left(1-\mathcal L[f](\tau)\right)
%\end{equation}
%and thus the two approaches are equivalent. 

The intuitive derivation for $\mathcal{R}^p_0$ follows from considering the rate at which new $S$--$I$ links are created.
From~\eqref{eq:closedeqSI}, and focusing on the single positive term on the right hand side, it follows that $S$--$I$ links are created at rate $\frac{\tau(n-1)}{n}\frac{[SS]}{[S]}$ which at time 
$t=0$ and with a vanishingly small initial number of infected nodes reduces to $\tau (n-1)$. Now, multiplying this by the average lifetime of an $S$--$I$ link, which is $\frac{1-\mathcal L[f_\mathcal{I}](\tau)}{\tau}$ \cite{konyv}, gives the desired threshold value in the limit of $[S] \rightarrow N$ at $t=0$.

Notice that while $\mathcal{R}_0$ depends on the expected value only, see \eqref{eq:mfstandardR0}, 
the pairwise reproduction number \eqref{eq:genR0p} uses the complete density function, thus the average length of the infectious period itself does not determine 
exactly the reproduction number. As a consequence, for an epidemic we have to know as precisely as possible the shape of the distribution. We shall analyse how the basic reproduction number \eqref{eq:genR0p}, which is not only an epidemic threshold but also determines the final size via \eqref{eq:finalsizegenimprel}, depends on the variance of the recovery time distribution. In \cite{biomat}, using gamma, lognormal and uniform distributions we showed that within each of those distribution families,
once the mean infectious period is fixed, smaller variance in the infectious period gives a higher reproduction number and consequently a more severe epidemic. Next we  generalize this result without restricting ourselves to special distributions.

\section*{Main Result: Relationship Between the Variance and the Reproduction Number}

\medskip

In this section we give some simple conditions which may guarantee that smaller variance induces higher pairwise reproduction number. We consider a random variable $\mathcal{I}$ corresponding to recovery times with probability density functions $f_{\mathcal{I}}(t)$, cumulative distribution function $F_{\mathcal{I}}(t)=\int_{0}^{t} f_{\mathcal{I}}(s) ds$ and we shall use the integral function of the CDF $\mathcal{F}_{\mathcal{I}}(t) :=  \int_{0}^{t} F_{\mathcal{I}}(s) ds$. Clearly, $\frac{d^2}{dt^2}\mathcal{F}_{\mathcal{I}}(t)=\frac{d}{dt}F_{\mathcal{I}}(t)=f_{\mathcal{I}}(t)$. Moreover, $F_{\mathcal{I}}(0)=\mathcal{F}_{\mathcal{I}}(0)=0.$

\begin{theorem}
	\label{generalimpact}
	Consider two random variables $\mathcal{I}_1$ and $\mathcal{I}_2$ such that
	\begin{equation}
		\label{eq:E}
		\mathbb{E}(\mathcal{I}_1)=\mathbb{E}(\mathcal{I}_2)<\infty,
	\end{equation}
	and 
	\begin{equation}
		\label{eq:V}
		\mathrm{Var}(\mathcal{I}_1)<\mathrm{Var}(\mathcal{I}_2)<\infty.	
	\end{equation}
	Assume that
	\begin{equation}
		\label{eq:M3alt}
		\lim\limits_{t\rightarrow\infty} t^3 f_{\mathcal{I}_j}(t) = 0, \quad j\in \{1,2\},
	\end{equation}
	and for all $t>0$, 
	\begin{equation}
		\label{eq:mainassumption}
		\mathcal{F}_{\mathcal{I}_1}(t)\neq\mathcal{F}_{\mathcal{I}_2}(t)
	\end{equation}	
	holds. If $\mathcal{I}_1$ and $\mathcal{I}_2$ represent the recovery time distribution, then for the corresponding reproduction numbers the relation $\mathcal{R}_{0,\mathcal{I}_1}^p>\mathcal{R}_{0,\mathcal{I}_2}^p$  holds.
\end{theorem}
\begin{proof}
	Using assumption~\eqref{eq:E}, we deduce
	\begin{eqnarray*}
		\int_{0}^{\infty} t\left(f_{\mathcal{I}_1}(t)-f_{\mathcal{I}_2}(t)\right) dt&=&\left[t(F_{\mathcal{I}_1}(t)-F_{\mathcal{I}_2}(t))\right]_0^\infty-\int_{0}^{\infty}(F_{\mathcal{I}_1}(t)-F_{\mathcal{I}_2}(t)) dt \nonumber\\
		&=&\lim\limits_{t\rightarrow\infty} t(F_{\mathcal{I}_1}(t)-F_{\mathcal{I}_2}(t))-\left[\mathcal{F}_{\mathcal{I}_1}(t)-\mathcal{F}_{\mathcal{I}_2}(t)\right]_0^\infty\nonumber\\
		&\stackrel{[*]}{=}&-\lim\limits_{t\rightarrow\infty}(\mathcal{F}_{\mathcal{I}_1}(t)-\mathcal{F}_{\mathcal{I}_2}(t))=0
	\end{eqnarray*}
	thus 
	\begin{equation}
		\label{eq:Ealt}
		\lim\limits_{t\rightarrow\infty}(\mathcal{F}_{\mathcal{I}_1}(t)-\mathcal{F}_{\mathcal{I}_2}(t))=0.
	\end{equation}
	To see $[*]$, i.e. $\lim\limits_{t\rightarrow\infty} t(F_{\mathcal{I}_1}(t)-F_{\mathcal{I}_2}(t))=0$,
	we need some algebraic manipulations:
	\begin{eqnarray*}
		\lim\limits_{t\rightarrow\infty} t(F_{\mathcal{I}_1}(t)-F_{\mathcal{I}_2}(t))&=&\lim\limits_{t\rightarrow\infty} \frac{F_{\mathcal{I}_1}(t)-F_{\mathcal{I}_2}(t)}{\frac{1}{t}}\stackrel{\mathrm{L'H}}{=}\lim\limits_{t\rightarrow\infty}\frac{f_{\mathcal{I}_1}(t)-f_{\mathcal{I}_2}(t)}{-\frac{1}{t^2}}\nonumber\\
		&=&-\lim\limits_{t\rightarrow\infty} t^2(f_{\mathcal{I}_1}(t)-f_{\mathcal{I}_2}(t)\stackrel{\eqref{eq:M3alt}}{=}0,
	\end{eqnarray*}
	where L'H refers to the L'Hospital rule. From assumption~\eqref{eq:V}, we have
	\begin{eqnarray*}
		\mathrm{Var}(\mathcal{I}_1)&=&\mathbb{E}(\mathcal{I}_1^2)-(\mathbb{E}(\mathcal{I}_1))^2<\mathbb{E}(\mathcal{I}_2^2)-(\mathbb{E}(\mathcal{I}_2))^2=\mathrm{Var}(\mathcal{I}_2)\nonumber\\
		&& \stackrel{\eqref{eq:E}}{\Rightarrow}\mathbb{E}(\mathcal{I}_1^2)<\mathbb{E}(\mathcal{I}_2^2),
	\end{eqnarray*}
	or equivalently $\int_{0}^{\infty}t^2(f_{\mathcal{I}_1}-f_{\mathcal{I}_2})dt<0$. We can carry out some calculations on the left-hand side of this inequality:
	\begin{eqnarray*}
		\int_{0}^{\infty}t^2(f_{\mathcal{I}_1}-f_{\mathcal{I}_2})dt &=& [t^2 (F_{\mathcal{I}_1}(t)-F_{\mathcal{I}_2}(t))]_0^\infty - 2 \int_{0}^{\infty} t(F_{\mathcal{I}_1}(t)-F_{\mathcal{I}_2}(t)) dt\nonumber\\
		&=&\lim\limits_{t\rightarrow\infty}t^2 (F_{\mathcal{I}_1}(t)-F_{\mathcal{I}_2}(t))-2 [t(\mathcal{F}_{\mathcal{I}_1}(t)-\mathcal{F}_{\mathcal{I}_2}(t))]_0^\infty\nonumber\\
		&+& 2 \int_{0}^{\infty} \mathcal{F}_{\mathcal{I}_1}(t)-\mathcal{F}_{\mathcal{I}_2}(t)dt\nonumber\\
		&\stackrel{[**]}{=}& -2 \lim\limits_{t\rightarrow\infty} t(\mathcal{F}_{\mathcal{I}_1}(t)-\mathcal{F}_{\mathcal{I}_2}(t))+ 2 \int_{0}^{\infty} \mathcal{F}_{\mathcal{I}_1}(t)-\mathcal{F}_{\mathcal{I}_2}(t)dt\nonumber\\
		&\stackrel{[**]}{=}& 2\int_{0}^{\infty} \mathcal{F}_{\mathcal{I}_1}(t)-\mathcal{F}_{\mathcal{I}_2}(t)dt,
	\end{eqnarray*}
	consequently 
	\begin{equation}
		\label{eq:Valt}
		\int_{0}^{\infty} \mathcal{F}_{\mathcal{I}_1}(t)-\mathcal{F}_{\mathcal{I}_2}(t)dt<0.
	\end{equation}
	To prove $[**]$, i.e. $\lim\limits_{t\rightarrow\infty}t^2 (F_{\mathcal{I}_1}(t)-F_{\mathcal{I}_2}(t))=\lim\limits_{t\rightarrow\infty} t(\mathcal{F}_{\mathcal{I}_1}(t)-\mathcal{F}_{\mathcal{I}_2}(t))=0$, we have
	\begin{eqnarray*}
		\lim\limits_{t\rightarrow\infty} t(\mathcal{F}_{\mathcal{I}_1}(t)-\mathcal{F}_{\mathcal{I}_2}(t))&=&\lim\limits_{t\rightarrow\infty} \frac{\mathcal{F}_{\mathcal{I}_1}(t)-\mathcal{F}_{\mathcal{I}_2}(t)}{\frac{1}{t}}
		\stackrel{\mathrm{L'H}}{=}\lim\limits_{t\rightarrow\infty}\frac{F_{\mathcal{I}_1}(t)-F_{\mathcal{I}_2}(t)}{-\frac{1}{t^2}}\nonumber\\
		&=& -\lim\limits_{t\rightarrow\infty}t^2 (F_{\mathcal{I}_1}(t)-F_{\mathcal{I}_2}(t))\nonumber\\
		&\stackrel{\mathrm{L'H}}{=}& \lim\limits_{t\rightarrow\infty}\frac{f_{\mathcal{I}_1}(t)-f_{\mathcal{I}_2}(t)}{\frac{2}{t^3}}=\frac{1}{2}\lim\limits_{t\rightarrow\infty} t^3(f_{\mathcal{I}_1}(t)-f_{\mathcal{I}_2}(t))\nonumber\\
		&\stackrel{\eqref{eq:M3alt}}{=}& 0.
	\end{eqnarray*}	
	Since $F_{\mathcal{I}}(t)\geq0, t\geq0$ and monotone increasing, the integral function of CDF $\mathcal{F}_{\mathcal{I}}(t)$ is monotone increasing and convex. Using \eqref{eq:mainassumption} and \eqref{eq:Valt}, we obtain
	\begin{equation}
		\label{eq:dominance}
		\mathcal{F}_{\mathcal{I}_1}(t)<\mathcal{F}_{\mathcal{I}_2}(t),
	\end{equation}
	for all $t>0$. Clearly, for $\mathcal{R}_{0,\mathcal{I}_1}^p>\mathcal{R}_{0,\mathcal{I}_2}^p$, it is enough to prove, that $\mathcal{L}[f_{\mathcal{I}_1}](\tau)<\mathcal{L}[f_{\mathcal{I}_2}](\tau)$, i.e. $\int_{0}^{\infty}e^{-\tau t}(f_{\mathcal{I}_1}(t)-f_{\mathcal{I}_2}(t))dt<0.$ First, we perform some algebraic manipulation on the left-hand side:
	\begin{eqnarray*}
		\int_{0}^{\infty}e^{-\tau t}(f_{\mathcal{I}_1}(t)-f_{\mathcal{I}_2}(t))dt&=&[e^{-\tau t}(F_{\mathcal{I}_1}(t)-F_{\mathcal{I}_2}(t))]_0^\infty\nonumber\\
		&& +\tau \int_{0}^{\infty}e^{-\tau t}(F_{\mathcal{I}_1}(t)-F_{\mathcal{I}_2}(t)) dt\nonumber\\
		&=& \tau [e^{-\tau t}(\mathcal{F}_{\mathcal{I}_1}(t)-\mathcal{F}_{\mathcal{I}_2}(t))]_0^\infty\nonumber\\
		&& +\tau^2 \int_{0}^{\infty}e^{-\tau t}(\mathcal{F}_{\mathcal{I}_1}(t)-\mathcal{F}_{\mathcal{I}_2}(t)) dt\nonumber\\
		&\stackrel{\eqref{eq:Ealt}}{=}& \tau^2 \int_{0}^{\infty}e^{-\tau t}(\mathcal{F}_{\mathcal{I}_1}(t)-\mathcal{F}_{\mathcal{I}_2}(t)) dt.
	\end{eqnarray*}
	In conclusion, we have 
	\begin{equation}
		\tau^2 \int_{0}^{\infty}e^{-\tau t}(\mathcal{F}_{\mathcal{I}_1}(t)-\mathcal{F}_{\mathcal{I}_2}(t)) dt \stackrel{\eqref{eq:dominance}}{<}0,  \label{remarkhoz}
	\end{equation}
	therefore $\mathcal{L}[f_{\mathcal{I}_1}](\tau)<\mathcal{L}[f_{\mathcal{I}_2}](\tau)$, which gives $\mathcal{R}_{0,\mathcal{I}_1}^p>\mathcal{R}_{0,\mathcal{I}_2}^p$.
\end{proof}

\begin{remark} While one can easily construct a specific example for which the technical condition \eqref{eq:M3alt} does not hold, it is satisfied by all epidemiologically meaningful distributions, since extremely long infectious periods do not occur in epidemics. It trivially holds for distributions with compact support, and even for power law distributions with finite variance.
	\end{remark}

\begin{corollary}
\label{cor:fsize}
Assume that the conditions of Theorem 1 hold. Then the infectious period distribution with smaller variance induces a larger epidemic outbreak. 
\end{corollary}
\begin{proof}
Let $z=s_\infty^{\frac{1}{n}}$. Then, \eqref{eq:finalsizegenimprel} can be written as
\begin{align}
(z-1)(n-1)&=\mathcal{R}^p_0\left(z^{n-1}-1\right), 
\end{align}
which, since we are interested in the root $z\in (0,1)$, simplifies to
\begin{align}
(n-1)&=\mathcal{R}^p_0\left(z^{n-2}+\dots+z+1\right). 
\end{align}
Clearly larger $\mathcal{R}^p_0$ results in smaller $z$, that means smaller $s_\infty$ thus larger epidemic. Combining this with Theorem 1 yields the result. 

\end{proof}

\section*{Relation to Stochastic Ordering and the work of Wilkinson and Sharkey}

\medskip

In a very recent work \cite{wilkinson2}, Wilkinson and Sharkey considered a general class of network based stochastic epidemic models, and proved a
monotonic relationship between the variability of the infectious period and the probability that the infection will spread to an arbitrary subset of the population by time $t$. Below we show that, while the work \cite{wilkinson2} was done in a different context, the main conclusion is very similar to our main result.
In \cite{wilkinson}, the variability was represented by the convex order of the distributions of infectious periods. Given two random variables $\mathcal{I}_1$ and $\mathcal{I}_2$ whose expectations exist, such that
	\begin{equation}
\mathbb{E}(\phi(\mathcal{I}_1)) \leq \mathbb{E}(\phi(\mathcal{I}_2)) \hbox{\quad  for all convex functions \quad } \phi: \mathbb{R} \to \mathbb{R},
\end{equation}
$\mathcal{I}_1$ is said to be smaller than $\mathcal{I}_2$ in the convex order, denoted by $\mathcal{I}_1 \leq_{cx} \mathcal{I}_2$, see monograph \cite{order} for a comprehensive description of various stochastic orders, their properties and relations.

\begin{theorem}
Assume that $\mathcal{I}_1 \leq_{cx} \mathcal{I}_2$ , and the technical condition \eqref{eq:M3alt} holds. Then, $\mathcal{R}_{0,\mathcal{I}_1}^p>\mathcal{R}_{0,\mathcal{I}_2}^p$  holds.
\end{theorem}
\begin{proof}
From the convexity of $\phi(x)=x$ and $\phi(x)=-x$, \eqref{eq:E} follows, and the convexity of $\phi(x)=x^2$ yields \eqref{eq:V}. From the convexity of $\phi_a(x)=(x-a)_+$, Theorem 3.A.1 in \cite{order} deduced that $\mathcal{I}_1 \leq_{cx} \mathcal{I}_2$ if and only if $\mathcal{F}_{\mathcal{I}_1}(t)\leq\mathcal{F}_{\mathcal{I}_2}(t)$ for all $t>0$.  Now instead of the strict inequality of \eqref{eq:dominance}, we have less or equal, but from \eqref{eq:V} the two functions are not identical, 
hence analogously to the proof of Theorem 1 we can conclude \eqref{remarkhoz}, which completes the proof.
\end{proof}

\begin{remark}
	\textbf{One can deduce Theorems 1 and 2 using \cite{wilkinson2}.} In \cite{wilkinson2}, the authors found the monotonic relationship between the variability of infectious periods and the final epidemic size, in a more general context of stochastic epidemics, that includes the pairwise models, by the means of convex ordering. According to Theorem 3.A.1b in \cite{order}, our condition \eqref{eq:mainassumption} implies that the two distributions considered in Theorem 1 are convex ordered, and then the main conclusion of Theorem 1 follows from combining \cite{wilkinson2} with the argument of Corollary 1 (monotonicity relationship between the pairwise reproduction number and the final epidemic size) . This also shows that Theorem 2 can be derived even without the technical condition, via \cite{wilkinson2}. 
\end{remark}

\begin{remark} \textbf{An example when \cite{wilkinson2} can not be applied but our methodology works.}

Let $\mathcal{I}_1 \sim \mathrm{Exp}(1)$ (exponential distribution with parameter $1$). Then $f_{\mathcal{I}_1}(t)=e^{-t}$, $F_{\mathcal{I}_1}(t)=1-e^{-t}$, $\mathcal{F}_{\mathcal{I}_1}(t)=t-1+e^{-t}$, $\mathbb{E}(\mathcal{I}_1)=1$, $\mathrm{Var}(\mathcal{I}_1)=1$, and $\mathcal{L}[f_{\mathcal{I}_1}](\tau)=1/(1+\tau)$.

Let $\mathcal{I}_2$ be the discrete random variable that takes the value $1-u>0$ with probability $0.5$, and the value $1+u$ with probability $0.5$, where $0<u<1$. Then, we have $\mathbb{E}(\mathcal{I}_2)=1$, $F_{\mathcal{I}_2}(t)=0$ for $t<1-u$, $0.5$ for $1-u \leq t <1+u$ and $1$ for $1+u \leq t$. Furthermore, $\mathrm{Var}(\mathcal{I}_2)=u^2<1$, $\mathcal{F}_{\mathcal{I}_2}(t)=0.5(t-(1-u))$ on $[1-u,1+u]$, and $\mathcal{L}[f_{\mathcal{I}_2}](\tau)=0.5( e^{-\tau (1-u)}+ e^{-\tau (1+u)})$.

Then, for $0<t<1-u$ we have $\mathcal{F}_{\mathcal{I}_1}(t)>0=\mathcal{F}_{\mathcal{I}_2}(t)$. 
However, at $t=1$, we have $\mathcal{F}_{\mathcal{I}_1}(1)= e^{-1} <0.5 u =\mathcal{F}_{\mathcal{I}_2}(1),$
whenever $u>2 e^{-1} \approx 0.736$. In light of Theorem 3.A.1b in \cite{order}, in this case the random variables $\mathcal{I}_1,\mathcal{I}_2$ are not convex ordered, thus \cite{wilkinson2} does not apply.

For sufficiently large $\tau$, we have $\mathcal{L}[f_{\mathcal{I}_1}](\tau)>\mathcal{L}[f_{\mathcal{I}_2}](\tau)$, hence  $\mathcal{R}_{0,\mathcal{I}_1}^p<\mathcal{R}_{0,\mathcal{I}_2}^p$, and the discrete random variable, which has the smaller variance, generates a larger epidemic outbreak. The pairwise reproduction number approach can be applied even in situations that are not covered by the convex order approach, as this simple example illustrates.
\end{remark}

\section*{Implications for Heterogeneous Degree Distributions}

\medskip

In a Configuration-Model network, given a random $S$--$I$ link, we expect the susceptible individual to have degree $k$ with probability proportional to $k[S_k]$ where $[S_k]$ is the number of susceptible individuals with degree $k$.

Repeating our earlier derivation of equation~\eqref{eq:genR0p} for $\mathcal{R}_0^p$ in the homogeneous network case, we anticipate that for fixed duration $\sigma$,
\begin{equation}
\mathcal{R}_0^p = \sum (k-1) (1-e^{-\tau \sigma}) \frac{k [S_k](0)}{\ave{k}N},
\end{equation}
where $\ave{k}$ is the average degree.

Extending this to the case of heterogeneous infection duration, we find
\begin{equation}
\mathcal{R}_0^p = (1 -\mathcal{L}[f_{\mathcal{I}}](\tau)) \frac{\sum_k k [S_k](0)}{N\ave{k}}.
\label{eq:hetR0p}
\end{equation}

It can be shown~\cite{Perc,kenah:second,newman2002spread} that the final number of degree $k$ individuals infected is given by 
\begin{equation}
[S_k]_\infty = [S_k]_0 \theta_\infty^k
\label{eq:Skinf}
\end{equation}
where the following implicit relation holds:
\begin{equation}
\theta_\infty = \mathcal{L}[f_{\mathcal{I}}](\tau) + (1-\mathcal{L}[f_{\mathcal{I}}](\tau)) \frac{\sum_k [S_k]_0\theta_\infty^{k-1}}{N\ave{k}}
\label{eq:Thetainf}
\end{equation}
Here $\theta_\infty$ is a per-edge measure of the probability of \emph{not} being infected.  So an initially susceptible individual with degree $k$ remains susceptible with probability $\theta_\infty^k$.  The role of $\theta_\infty$ is the same as $s_\infty^{1/n}$ in Eq.~\eqref{eq:homogeneous_final}.

Note that in $\mathcal{R}_0^p$, the terms capturing the distribution of infection durations separate from the terms capturing the distribution of degrees.  The ordering of $\mathcal{R}_0^p$ as the infection duration distribution changes is independent of the degree distribution.  So the ordering of $\mathcal{R}_0^p$ is the same as found in the regular networks.  The final size depends monotonically on the Laplace transform of $f_{\mathcal{I}}$, and so the results about the ordering of final sizes in regular networks carry over to heterogeneous networks as well.

\section*{Discussion}

\medskip

The role of the shape of the distribution of infectious periods in disease spread has been in the interest of modellers for some time \cite{wallinga}. Our previous works already indicated that for pairwise models of network epidemic, not only the mean,
but higher order properties of the distribution of the recovery times have an impact on the outcome of the epidemic. We derived useful threshold quantities for non-Markovian recovery in \cite{prl}. In \cite{biomat}, we showed that for particular distribution families (typically two parameter families such as gamma, lognormal, and uniform distribution), smaller variance leads to higher reproduction number within the same family when the mean is fixed. Our new result in this study allows us to make comparisons between distributions of different kinds. To show the usefulness of Theorem 1,    
as an example, we consider $\mathcal{I}_1 \sim \mathrm{Exp}(\gamma)$ and $\mathcal{I}_2 \sim \mathrm{Fixed}\left(\frac{1}{\gamma}\right)$, i.e. $f_{\mathcal{I}_1}(t)=\gamma e^{-\gamma t}, t\geq 0$ and $f_{\mathcal{I}_2}(t)=\delta\left(t-\frac{1}{\gamma}\right)$, where $\delta(t)$ denotes the Dirac delta function. Clearly, we obtain $\mathcal{F}_{\mathcal{I}_1}(t)=t+\frac{1}{\gamma}e^{-\gamma t}-\frac{1}{\gamma}$ and $\mathcal{F}_{\mathcal{I}_2}(t)=(t-\frac{1}{\gamma})_+$, thus there is no $t_0>0$, such that $\mathcal{F}_{\mathcal{I}_1}(t_0)=\mathcal{F}_{\mathcal{I}_2}(t_0)$. Since $\mathbb{E}(\mathcal{I}_1)=\mathbb{E}(\mathcal{I}_2)=\frac{1}{\gamma}$, $\frac{1}{\gamma^2}=\mathrm{Var}(\mathcal{I}_1)>\mathrm{Var}(\mathcal{I}_2)=0$ and the other conditions of Theorem \ref{generalimpact} are satisfied, we find $\mathcal{R}_{0,\mathcal{I}_1}^p<\mathcal{R}_{0,\mathcal{I}_2}^p$. 

\begin{table}[h]
%\label{tab:distparammeanvar}
 \begin{tabular}{||c c c c||} 
 \hline
 Distribution & Parameters & Mean & Variance \\ [0.5ex] 
 \hline\hline
 Fixed & 3/2 & 3/2 & 0 \\ 
 \hline
 Uniform & U(1,2) & 3/2 &  1/12=0.08(3)\\
 \hline
 Gamma &  \text{scale} =0.5, \ \text{shape} = 3
 & 3/2 &  0.75\\  
 \hline
 Exponential & 2/3 & 3/2 & 9/4=2.25 \\
 \hline
 Lognormal & $\sigma = 1$,  \ $\mu=\ln(3/2)-1/2$ & 3/2 & 3.866  \\
 \hline
 Weibull & \text{scale} = 1, \ \text{shape} = 0.6014 & 3/2 & 6.914 \\ 
 \hline
\end{tabular}
\vspace{0.5cm}
\caption{Details of all the distributions of the infection times used for the explicit stochastic network simulations.}\label{tab:distparammeanvar}
\end{table}

We have carried out extensive numerical simulations to test the final epidemic size formula \eqref{eq:finalsizegenimprel}, with $\mathcal{R}_0^p$ taken from \eqref{eq:hetR0p},  
for Fixed, Uniform, Gamma, Exponential, Log-normal, Weibull distributed infection times on regular (see Fig.~\ref{fig:regularsize}, Fig.~\ref{fig:regularevol}), Erd\H{o}s-R\'enyi (see Fig.~\ref{fig:ERsize}, Fig.~\ref{fig:ERevol}) 
and truncated scale-free (see Fig.~\ref{fig:PLsize}, Fig.~\ref{fig:PLevol}) networks. It is worth noting that the same final size relation can be obtained by combining equations \eqref{eq:Skinf}, 
\eqref{eq:Thetainf} and that for $\mathcal{R}_0^p$ for the heterogenous degree distributions. The agreement between the analytical final epidemic size and explicit 
stochastic network simulations is excellent for all distributions and networks. The parameters, mean and variance of the distributions are given in Table~\ref{tab:distparammeanvar}.

Several observations can be made. In Figs.~\ref{fig:regularsize}, \ref{fig:ERsize} and \ref{fig:PLsize} one can note that the epidemic threshold depends heavily on the distribution 
of the infectious period. While all distributions have the same mean, they differ in terms of their variance. In fact, the variance of the distributions are ordered as shown in Table~\ref{tab:distparammeanvar}.
%\begin{subequations}
%\label{eq:varorder}
%\begin{align}
%Var(\text{fixed})=0 \le Var(\text{Uniform}) \le Var(\text{exponential}) \le \\
%Var(\text{lognormal}) \le Var(\text{weibull}) \le Var(\text{Gamma}).
%	\end{align}
%\end{subequations}
Based on Theorem~\ref{generalimpact} and Corollary~\ref{cor:fsize} we know that for a fixed transmission rate $\tau$ and for infectious period distributions with the same mean, 
the distribution with the higher variance will lead to a smaller $\mathcal{R}_0^p$ and hence smaller attack rate. This confirms that the ordering of the variances in 
Table \ref{tab:distparammeanvar} is reflected accurately in all attack rate versus $\tau$ plots. Moreover, the insets in Figs.~\ref{fig:regularsize}, \ref{fig:ERsize} and 
\ref{fig:PLsize} 
shows that the final epidemic size relation in terms of $\mathcal{R}_0^p$ is universal, independently of how the infectious periods are distributed.
For the truncated scale-free networks in Fig.~\ref{fig:PLsize}, the attack rate behaves differently but the general analytical final epidemic size relation remains extremely 
accurate. Obviously high degree heterogeneity leads to large variance and this makes the value of $\mathcal{R}_0^p$ to be large and above threshold even for small values of 
$\tau$.

\section*{Conclusion}

\medskip

Figures~\ref{fig:regularevol}, \ref{fig:ERevol} and \ref{fig:PLevol} show the initial growth of the epidemic. The relation between variance and attack rate seems to  
translate into a straightforward association between variance and initial growth rate. Namely, distributions with higher variance leads to slower initial growth. This is not always 
the case since $\mathcal{R}_0^p$ is a generation rather than time based measure. However, here the mean of the distributions and the transmission rates are identical and thus 
the ordering seems to carry through.

We can offer an intuitive explanation of our result.  A key factor determining how many infections occur is the proportion of SI edges that eventually transmit.  If we have $M$ edges where $M$ is large, with an average infection duration $D$ and transmission rate $\tau$, then the expected number of transmission events to occur is $D M \tau$, but only the first transmission event per edge has any impact.  Those edges in which the infection duration is longer will tend to have more transmission events, while those with shorter duration are more likely to have no transmission events.  Increasing the variance in duration tends to increase the concentration of the transmission events into a smaller set of edges, resulting in fewer successful transmissions, and conversely, decreasing the variance redistributes some transmission events from edges which have already transmitted to those edges which have not.

As next steps one could consider the extension of $\mathcal{R}_0^p$ and the final size formula for epidemics where both the infection and transmission processes are non-Markovian. Such results already exist \cite{Neil} but there an EBCM was used. It would also be appropriate to explore the applicability of this newly introduced pairwise reproduction number given that it lent itself to derive a number of analytical results and it fits with the network and contact concepts. In particular one would explore how could this be measured in practice and how does its value translate into control measures.
%Finally we note that the cal

\section*{Abbreviations}

\quad\!\!\! CDF: cumulative distribution function

EBCM: edge-based compartmental model

ODE: ordinary differential equation

SIR: susceptible-infected-removed

%\nocite{oreg,schn,pond,smith,marg,hunn,advi,koha,mouse}

%%%%%%%%%%%%%%%%%%%%%%%%%%%%%%%%%%%%%%%%%%%%%%
%%                                          %%
%% Backmatter begins here                   %%
%%                                          %%
%%%%%%%%%%%%%%%%%%%%%%%%%%%%%%%%%%%%%%%%%%%%%%

\begin{backmatter}

\section*{DECLARATIONS}
  
\medskip

\section*{Author's contributions}

The study was conceived by GR. The main result and its proof was found by ZsV. Network epidemic methodology was provided by IZK.
Numerical studies were done by JCM. All authors contrituted to the writing of the manuscript.

\section*{Acknowledgements}
  
  This article is a significantly extended version of R\"ost G., Kiss I. Z., Vizi Z., Variance of Infectious Periods and Reproduction Numbers for Network Epidemics with Non-Markovian Recovery, Progress in Industrial Mathematics at ECMI 2016. 
  
\section*{Funding}
 
  ZsV was supported by the EU-funded Hungarian grant EFOP-3.6.2-16-2017-00015. GR was supported by Hungarian National Research Fund Grant NKFI FK 124016 and MSCA-IF
  748193. JCM was supported by Global Good.

\section*{Availability of Data and Materials} 

Not applicable. The article contains no data.

\section*{Competing Interests} 

None of the authors have any competing interest.

%%  After submission of the .TEX file,                     %%
%%  you will be prompted to submit your .BBL file.         %%

% if your bibliography is in bibtex format, use those commands:
%\bibliographystyle{bmc-mathphys} % Style BST file (bmc-mathphys, vancouver, spbasic).
%\bibliography{bmc_article}      % Bibliography file (usually '*.bib' )
% for author-year bibliography (bmc-mathphys or spbasic)
% a) write to bib file (bmc-mathphys only)
% @settings{label, options="nameyear"}

\newpage

\begin{figure}[h!] 
	\includegraphics[width=0.8\textwidth]{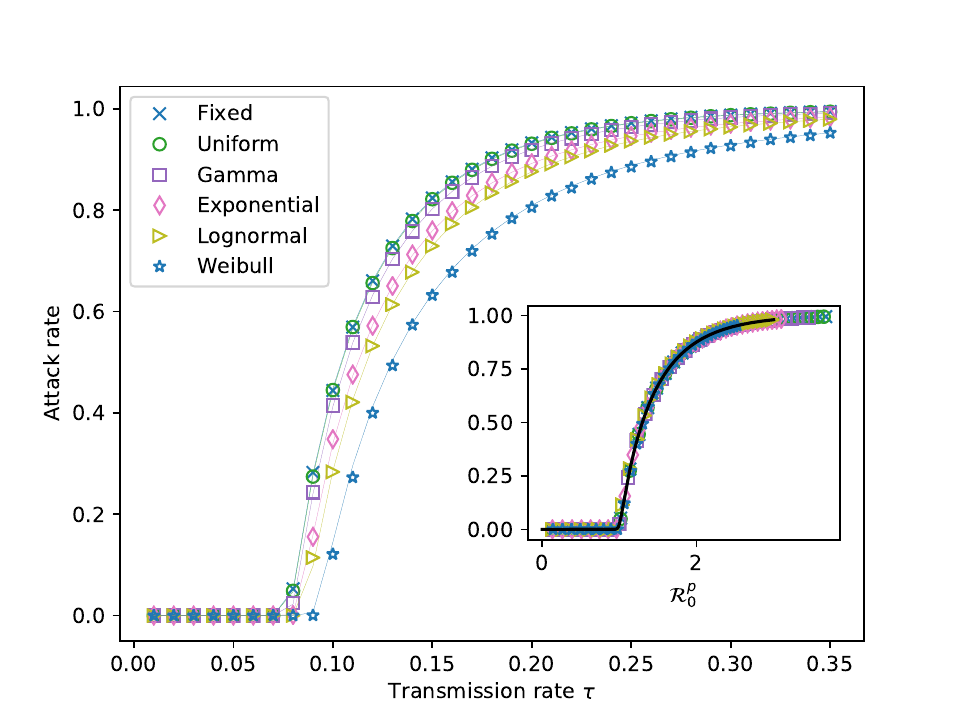}
	\caption{\csentence{Epidemic sizes in a regular network.}  We consider the outbreak sizes in a random network with $10^6$ nodes all having degree $k=10$.  We take distributions of infection duration having mean $3/2$ and plot the final proportion infected given different transmission rates $\tau$.  The inset shows that all final epidemic sizes collapse on a universal curve when using $\mathcal{R}_0^p$ as the horizontal axis. The parameters, mean and variance of the distributions are given in Table~\ref{tab:distparammeanvar}.
 }
    \label{fig:regularsize}
\end{figure}

\begin{figure}[h!] 
\includegraphics[width=0.9\textwidth]{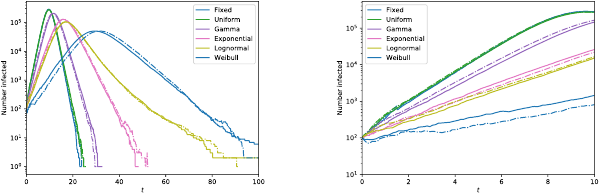}
	\caption{\csentence{Dynamics of epidemics in regular network.}
		Taking $\tau = 0.15$, we show two epidemic curves for each distribution from Fig.~\ref{fig:regularsize}.  On the right we see the early dynamics.  }
		\label{fig:regularevol}
\end{figure}

\begin{figure}[h!] 
	\includegraphics[width=0.8\textwidth]{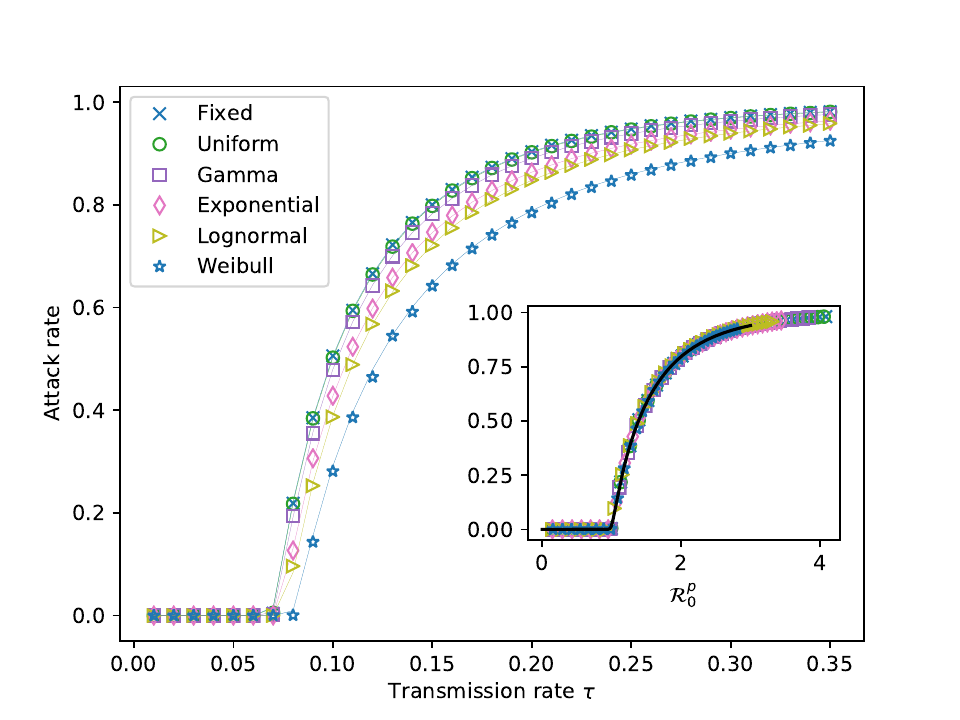}
    	\caption{\csentence{Epidemic sizes in an Erd\H{o}s--R\'{e}nyi network.}
		We look at epidemics in an Erd\H{o}s--R\'{e}nyi network with $10^6$ nodes and average degree $10$.  The results are similar to Fig.~\ref{fig:regularsize}.  Using the heterogeneous $\mathcal{R}_0^p$ all curves collapse on a universal curve. The parameters, mean and variance of the distributions are given in Table~\ref{tab:distparammeanvar}.}
\label{fig:ERsize}
\end{figure}

\begin{figure}[h!] 
\includegraphics[width=0.9\textwidth]{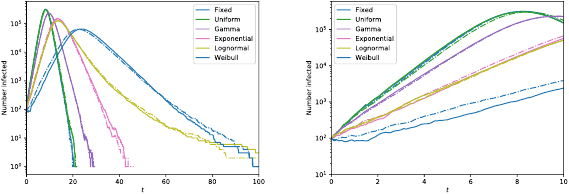}
	\caption{\csentence{Dynamics of epidemics in Erd\H{o}s--R\'{e}nyi network.}
		Taking $\tau = 0.15$, we show two epidemic curves for each distribution from Fig.~\ref{fig:ERsize}.  On the right we see the early dynamics.  }
		\label{fig:ERevol}
\end{figure}

\begin{figure}[h!] 
	\includegraphics[width=0.8\textwidth]{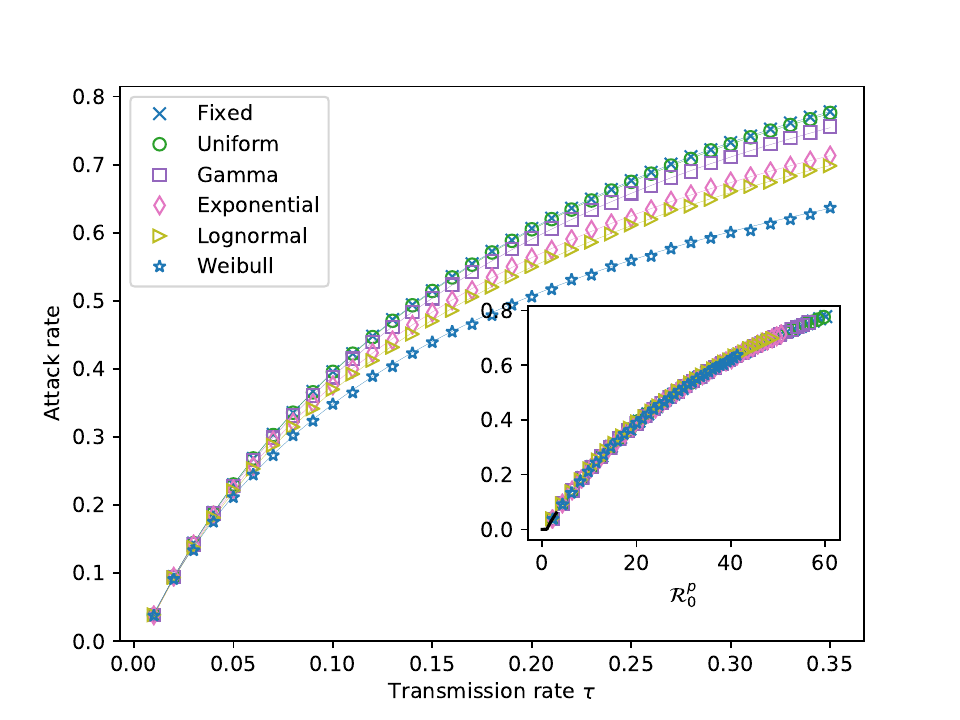}
    	\caption{\csentence{Epidemic sizes in a scalefree network.}
		We look at epidemics in a truncated scale free network with $10^6$ nodes having minimum degree $2$ and maximum degree $954$ and each degree $k$ assigned with probability proportional to $k^{-2}$.  This yields an average degree of approximately $10$.  Epidemics exist even at very small $\tau$, and $\mathcal{R}_0^p$ is significantly larger than in the other networks. Using the heterogeneous $\mathcal{R}_0^p$ all curves collapse on a universal curve. The parameters, mean and variance of the distributions are given in Table~\ref{tab:distparammeanvar}.}
\label{fig:PLsize}
\end{figure}

\begin{figure}[h!] 
	\includegraphics[width=0.9\textwidth]{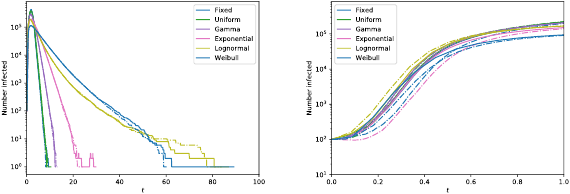}
	\caption{\csentence{Dynamics of epidemics in scalefree network.}
		Taking $\tau = 0.15$, we show two epidemic curves for each distribution from Fig.~\ref{fig:PLsize}.  On the right we see the early dynamics.  There is significant heterogeneity in the early growth, even for the same distribution.  This is because the timing of rare infections to the highest degree nodes plays a significant role even in networks of $10^6$ nodes.}
\label{fig:PLevol}		
\end{figure}
%%%%%%%%%%%%%%%%%%%%%%%%%%%%%%%%%%%
%%                               %%
%% Figures                       %%
%%                               %%
%% NB: this is for captions and  %%
%% Titles. All graphics must be  %%
%% submitted separately and NOT  %%
%% included in the Tex document  %%
%%                               %%
%%%%%%%%%%%%%%%%%%%%%%%%%%%%%%%%%%%

%%
%% Do not use \listoffigures as most will included as separate files

%%%%%%%%%%%%%%%%%%%%%%%%%%%%%%%%%%%
%%                               %%
%% Tables                        %%
%%                               %%
%%%%%%%%%%%%%%%%%%%%%%%%%%%%%%%%%%%

%% Use of \listoftables is discouraged.
%%

\end{backmatter}
\end{document}